\documentclass[%
 reprint
]{revtex4-2}
\usepackage{iceberg}
\usepackage{quiver}

\begin{document}

\title{A lower bound on the overhead of addressable surgery on sparse stabiliser codes}

\author{Nou\'edyn Baspin}
\affiliation{%
 Iceberg Quantum\\
 Sydney, Australia
}%
\author{{\scriptsize\texttt{nouedyn@iceberg-quantum.com}}}

\begin{abstract}
We show that, for stabilisers codes, there always exists a logical subspace requiring $\tilde{\Omega}(k^2)$ ancillary qubits to be measured with a sparse surgery ancilla. For CSS codes the hard subspace can be taken to be generated by all-X, or all-Z operators. As a result, for codes with large $k$, constant overhead ancillas can only realise a small fraction of possible measurements. 
\end{abstract}

\maketitle

The prospect of practical quantum computers requires the ability to record the outcome of a computation, once completed. However, the process of reading information out of a computer is made challenging by the presence of noise; as measurements used to estimate the computational state can be irreversibly corrupted. 

In quantum error correction, this problem corresponds to performing a reliable logical measurement on an encoded state. 
Early efforts to solve this issue led to the Shor and the Steane measurement schemes \cite{shor1990fault, steane1997active, Huang2023}; but with both suffering from complex state-preparation requirements, they incur non-trivial space and time overheads.

Lattice surgery \cite{bombin2009measurements, Horsman2012}, allows for greatly reducing these operational costs. While originally only applicable to surface codes, then topological codes, the progress toward generalised lattice surgery \cite{Cohen2022, Williamson2024, cowtan2025parallel} seems to hint at the possibility of low overhead measurements for all QLDPC codes. 

A crucial modality in the realisation of surgery, is whether multiple logicals operators can be measured simultaneously, which is often referred to as \emph{parallel} surgery \cite{cowtan2025parallel}. 
A construction of parallel surgery is given in \cite{cowtan2025parallel}, which allows for measuring any set of $t$ commuting logicals on an $n$-qubit QLDPC code, using at most 
\begin{equation} \label{eq:parallel}
    N_\star \leq \tilde{O}(tn)
\end{equation}
additional ancilla qubits. The main contribution of our work is a no-go theorem showing that some measurement incur an incompressible overhead, i.e
\begin{equation}
    N_\star \geq \tilde{\Omega}(k^2)
\end{equation}
In particular, \Cref{eq:parallel} is almost tight for constant rate codes. In which case, taking $t \propto n$, \Cref{thm:k-squared} gives $N_\star = \Theta(n^2)$.

 In full generality, our bound targets addressability, i.e. the capacity to exactly prescribe which degrees of freedom are read. Such a high degree of control can be expressed as the (large) set $\mathcal{F}$ of all the measurements allowed by the surgery scheme. Our result, \Cref{lem:main}, then reads
\begin{equation}
        N_\star \geq \Omega\left( \frac{\log |\mathcal{F}|}{ \log(n)} \right)
\end{equation}

The larger $|\mathcal{F}|$, or the more addressable, the worse is the space cost. We note that this bound seems to point towards a space-time tradeoff: in order to reduce the space overhead, one has to limit the number of logicals that can be measured at every surgery operation (or logical time step), and therefore increase the number of rounds of surgery needed. 

We conclude by discussing the the relationship of this bound with \cite{cowtan2026nearoptimal}. In that work, the authors give an almost-linear overhead surgery construction for a logical subspace, provided the existence of a suitable subcode. As a consequence of \Cref{thm:k-squared}, these subcodes happen to be in limited supply, at least for codes satisfying $k \gg \sqrt{n}$; and thus this method cannot accommodate all combinations of logicals of the same type -- which was also discussed in \cite{cowtan2026nearoptimal}.

\section{Lower bound}
We represent a stabiliser code $C$ on $n$ qubits and $m$ stabilisers as a symplectic complex \cite{yuan2026noncss}:
\begin{equation}
        C=S^{C} \xrightarrow{\sigma^C} P^C\xrightarrow{\hat{\sigma}^C} \bar{S}^{C}
\end{equation}
With $P^C \cong \mathbb{F}_2^{2n}$ representing to the space of physical Pauli operators, while $S^{C} \cong \bar{S}^{C} \cong \mathbb{F}_2^{m}$ represent the stabilisers. 
We also use the following grading
\begin{equation}
        C=C_2 \to C_1 \to C_0
\end{equation}
The space $\ker \hat{\sigma}^C$ is called the \emph{logicals} of $C$, which are classified by equivalence up to stabilisers. The code $C$ has then $2k = \dim H_1(C)$ inequivalent logicals.

Performing surgery on $C$, with a CSS ancilla $D = X^D\to Q^D_X\to \bar{Z}^{D}$, and associated grading $D = D_1 \to D_0 \to D_{-1}$ can then be formally expressed as an instance of a symplectic embedding \cite{yuan2026noncss}.
\begin{equation} \label{eq:symplectic-embedding}
    \begin{tikzcd}
	&& {X^D} & {Q^D_X} & {\bar{Z}^D} \\
	& {S^C} & {P^C} & {\bar{S}^C} \\
	{Z^D} & {Q_Z^D} & {\bar{X}^D}
	\arrow[from=1-3, to=1-4]
	\arrow[from=1-3, to=2-3]
	\arrow["p"'{pos=0.3}, color={rgb,255:red,92;green,214;blue,92}, curve={height=6pt}, dashed, from=1-3, to=3-2]
	\arrow[from=1-4, to=1-5]
	\arrow[from=1-4, to=2-4]
	\arrow["{p^\top}"{pos=0.7}, color={rgb,255:red,92;green,214;blue,92}, curve={height=6pt}, dashed, from=1-4, to=3-3]
	\arrow[from=2-2, to=2-3]
	\arrow[from=2-2, to=3-2]
	\arrow[from=2-3, to=2-4]
	\arrow[from=2-3, to=3-3]
	\arrow[from=3-1, to=3-2]
	\arrow[from=3-2, to=3-3]
\end{tikzcd}
\end{equation}
Where $(g,h)$ is a chain map, $p$ is a defect map satisfying $\hat{g}g={\partial^D}^\top p+ p^{\top} \partial^D $, and $H_0(D)=0$. We say that a map $\partial$ is $\omega$-sparse if both its row and column weight are upper bounded by $\omega$.

\begin{definition}[\cite{yuan2026noncss}]\label{def:measurement}
    The complex depicted in \Cref{eq:symplectic-embedding} measures the logical subspace $\mathfrak{l} = [g] (H_1(D)) \subset H_1(C)$.
\end{definition}

Due to the stochastic nature of quantum error correction, it is, in most cases, impossible to know ahead-of-computation which (commuting) subspace $\mathfrak{l}$ will have to be measured. We thus isolate a subset $\mathcal{F}$ of subspaces of $H_1(D)$ that will be measured during the computation. For each $\mathfrak{l} \in \mathcal{F}$, we assume there exists an ancillary complex $D^{(\mathfrak{l})}$ satisfying \Cref{def:measurement}. 

\begin{lemma} \label{lem:main}
    Let $\mathcal{F}$ a set of subspaces of $H_1(C)$, and $\{(D^{\mathfrak{l}},g^{\mathfrak{l}}) : \mathfrak{l} \in \mathcal{{F}}\}$ such that:
    \begin{enumerate}
        \item $(D^{\mathfrak{l}}, g^{\mathfrak{l}})$ measures $\mathfrak{l}$ in the sense of \Cref{def:measurement},
        \item $\partial_1^{D^{\mathfrak{l}}}$ and $g^{\mathfrak{l}}$ are $\omega$-sparse,
        \item $\dim D^{\mathfrak{l}}_1 + \dim D^{\mathfrak{l}}_0 \leq N_\star$,
    \end{enumerate}
    then 
     \begin{equation}
        N_\star \geq \Omega\left( \frac{\log |\mathcal{F}|}{\omega \log(n)} \right)
    \end{equation}
\end{lemma}

\begin{proof}
    We first perform a convenient manipulation that reduces the number of $\mathfrak{l}$-parametrised maps involved:
    \begin{align}
        \mathfrak{l} &= [g^{\mathfrak{l}}] H_1(D^{\mathfrak{l}}) \\
        &= (g^{\mathfrak{l}} \ker \partial_1^{D^{\mathfrak{l}}}) / \im \partial_2^C
    \end{align}
    We see then that the cardinality of $\mathcal{F}$ is upper bounded by the number of valid pairs $(D^{\mathfrak{l}}, g^{\mathfrak{l}})$. Let us fix $d_0, d_1$, and count the number of pairs for which
    \begin{equation}
        D^{\mathfrak{l}} \in \mathbb{F}_2^{d_0 \times d_1 \ }, g^{\mathfrak{l}} \in \mathbb{F}_2^{2n \times d_1\ }
    \end{equation}
    
    Using \Cref{lem:combinatorial-bound}, we can bound the number of such matrices by $(1+d_0)^{\omega d_1}$ and $(1+2n)^{\omega d_1}$ respectively, hence the total number of possible pairs is upper bounded by $((1+d_0)(1+2n))^{\omega d_1} \leq ((1+N_\star)(1+2n))^{\omega N_\star}$.

    It remains to count the number of pairs $(d_0, d_1)$ such that $d_0 + d_1 \leq N_\star$. Note that fixing $d_0 = N_\star$, we have only $d_1 = 0$ as a choice; then if $d= N_\star-1$ we have choices $\{0,1\}$, etc. This gives at most $(1+N_\star)$ choices. We conclude
    \begin{equation}
        |\mathcal{F}| \leq (1+N_\star)^2 ((1+N_\star)(1+2n))^{\omega N_\star} 
    \end{equation}
    or, 
    \begin{equation}
        N_\star \geq \frac{\log |\mathcal{F}|}{2 + \omega\log(1+2n) + \omega \log(1+N_\star)}
    \end{equation}
    We can eliminate the $\log(1+N_\star)$ term by noting that, unless $N_\star \geq \log |\mathcal{F}|$,in which case the claim is proved, we have:
    \begin{equation}
        N_\star \geq \frac{\log |\mathcal{F}|}{2 + \omega\log(1+2n) + \omega \log(1+\log |\mathcal{F}|)}
    \end{equation}
    Then note that $\mathcal{F}$ can be upper bounded by the number of combinations of $2k$ vectors with $2k$ entries over $\mathbb{F}_2$, hence $\log |\mathcal{F}| \leq 4k^2 \leq 4n^2$. 
    After further simplification we obtain
    \begin{equation}
        N_\star \geq \Omega\left( \frac{\log |\mathcal{F}|}{\omega \log(n)} \right)
    \end{equation}
\end{proof}

Ideally, we would like to be able to measure all possible subspaces of $H_1(C)$, to allows for complete addressable parallel surgery. However the above theorem imposes a space overhead for that task. In fact we show this overhead occurs even if one only desires to measure all subspaces of dimension exactly $k/2$. We achieve this by picking a polarisation of $H_1(C) \cong \mathcal{L}_X \oplus  \mathcal{L}_Z \cong \mathbb{F}_2^k \oplus \mathbb{F}_2^k$, with $\mathcal{L}_X$ corresponding to a choice of basis for $X$ logical operators, and similarly for $\mathcal{L}_Z$. With this choice of basis we can guarantee that all the logicals in a given $\mathfrak{l} \in \mathcal{F}$ do commute and therefore constitute valid measurements.

\begin{theorem} \label{thm:k-squared}
    Let $\mathcal{F}$ be the set of all possible measurements on a code $C$, and $\{(D^{\mathfrak{l}},g^{\mathfrak{l}}) : \mathfrak{l} \in \mathcal{{F}}\}$ a corresponding set of surgery ancillas such that:
    \begin{enumerate}
        \item $(D^{\mathfrak{l}}, g^{\mathfrak{l}})$ measures $\mathfrak{l}$ in the sense of \Cref{def:measurement},
        \item $\partial_1^{D^{\mathfrak{l}}}$ and $g^{\mathfrak{l}}$ are $\omega$-sparse,
        \item $\dim D^{\mathfrak{l}}_1 + \dim D^{\mathfrak{l}}_0 \leq N_\star$,
    \end{enumerate}
    then 
     \begin{equation}
        N_\star \geq \Omega\left( \frac{k^2}{\omega \log(n)} \right)
    \end{equation}
\end{theorem}
\begin{proof}
    Write $\mathcal{F}'$ the set of all $k/2$-dimensional subspaces of $\mathcal{L}_X$. We have $\mathcal{F}' \subset \mathcal{F}$. 
    We then apply \Cref{lem:main}, with the lower bound on $|\mathcal{F}| \geq |\mathcal{F}'|$ from \Cref{lem:number-subspaces} using $r = k/2, l = k = \dim \mathcal{L}_X$.
\end{proof}

\section{Auxillary lemmas}

\begin{lemma}\label{lem:combinatorial-bound}
    There are at most $(1+s)^{\omega t}$ $\omega$-sparse matrices of dimension $\mathbb{F}_2^{s \times t}$
\end{lemma}
\begin{proof}
    Denote $A$ such a matrix. We write it:
    \begin{equation}
        \begin{pmatrix}
            a_1 & a_2 & \dots & a_t
        \end{pmatrix}
    \end{equation}
    Where each $a_i \in \mathbb{F}_2^{s}$ is a binary vector of size $s$. If $|a_i| = \omega_i$, then there are $\binom{s}{\omega_i}$ choices for $a_i$. Considering that all choices up to weight $\omega$ are allowed, this gives
    \begin{equation}
        \sum_{\omega'=0}^{\omega} \binom{s}{\omega'} \leq (1+s)^{\omega}
    \end{equation}
    Where the inequality comes from the binomial theorem. With $t$ such vectors, we conclude there $(1+s)^{\omega t}$ choices for $A$.
    
\end{proof}

\begin{lemma}[\cite{koetter2007coding}] \label{lem:number-subspaces}
    The number of $r$-dimensional subspaces of $\mathbb{F}_2^l$ is lower bounded by $2^{r(l-r)}$.
\end{lemma}

\bibliography{refs}

@article{steane1997active,
  title = {Active Stabilization, Quantum Computation, and Quantum State Synthesis},
  author = {Steane, A. M.},
  journal = {Phys. Rev. Lett.},
  volume = {78},
  issue = {11},
  pages = {2252--2255},
  numpages = {0},
  year = {1997},
  month = {Mar},
  publisher = {American Physical Society},
  doi = {10.1103/PhysRevLett.78.2252},
  url = {https://link.aps.org/doi/10.1103/PhysRevLett.78.2252}
}

@INPROCEEDINGS{shor1990fault,
  author={Shor, P.W.},
  booktitle={Proceedings of 37th Conference on Foundations of Computer Science}, 
  title={Fault-tolerant quantum computation}, 
  year={1996},
  volume={},
  number={},
  pages={56-65},
  doi={10.1109/SFCS.1996.548464}}

@article{Cohen2022,
  author = {Lawrence Z. Cohen  and Isaac H. Kim  and Stephen D. Bartlett  and Benjamin J. Brown },
  title = {Low-overhead fault-tolerant quantum computing using long-range connectivity},
  journal = {Science Advances},
  volume = {8},
  number = {20},
  pages = {eabn1717},
  year = {2022},
  doi = {10.1126/sciadv.abn1717},
  URL = {https://www.science.org/doi/abs/10.1126/sciadv.abn1717},
  eprint = {https://www.science.org/doi/pdf/10.1126/sciadv.abn1717}
}

@misc{Williamson2024,
  title={Low-overhead fault-tolerant quantum computation by gauging logical operators}, 
  author={Dominic J. Williamson and Theodore J. Yoder},
  year={2024},
  eprint={2410.02213},
  archivePrefix={arXiv},
  primaryClass={quant-ph},
  url={https://arxiv.org/abs/2410.02213}, 
}

@misc{cowtan2025parallel,
  title={Parallel logical measurements via quantum code surgery},
  author={Cowtan, Alexander and He, Zhiyang and Williamson, Dominic J and Yoder, Theodore J},
  eprint={2503.05003},
  archivePrefix={arXiv},
  primaryClass={quant-ph},
  url={https://arxiv.org/abs/2503.05003}, 
  year={2025}
}

@article{Huang2023,
  title = {Homomorphic Logical Measurements},
  author = {Huang, Shilin and Jochym-O'Connor, Tomas and Yoder, Theodore J.},
  journal = {PRX Quantum},
  volume = {4},
  issue = {3},
  pages = {030301},
  numpages = {14},
  year = {2023},
  month = {Jul},
  publisher = {American Physical Society},
  doi = {10.1103/PRXQuantum.4.030301},
  url = {https://link.aps.org/doi/10.1103/PRXQuantum.4.030301}
}

@INPROCEEDINGS{koetter2007coding,
  author={Koetter, Ralf and Kschischang, Frank R.},
  booktitle={2007 IEEE International Symposium on Information Theory}, 
  title={Coding for Errors and Erasures in Random Network Coding}, 
  year={2007},
  volume={},
  number={},
  pages={791-795},
  doi={10.1109/ISIT.2007.4557321}}

@article{Horsman2012,
   title={Surface code quantum computing by lattice surgery},
   volume={14},
   ISSN={1367-2630},
   url={http://dx.doi.org/10.1088/1367-2630/14/12/123011},
   DOI={10.1088/1367-2630/14/12/123011},
   number={12},
   journal={New Journal of Physics},
   publisher={IOP Publishing},
   author={Horsman, Dominic and Fowler, Austin G and Devitt, Simon and Meter, Rodney Van},
   year={2012},
   month=dec, pages={123011} }

@misc{yuan2026noncss,
      title={Non-CSS Quantum Code Embedding}, 
      author={Andrew C. Yuan and Nouédyn Baspin},
      year={2026},
      eprint={2608.16995},
      archivePrefix={arXiv},
      primaryClass={quant-ph},
      url={https://arxiv.org/abs/2608.16995}, 
}

@article{bombin2009measurements,
   title={Quantum measurements and gates by code deformation},
   volume={42},
   ISSN={1751-8121},
   url={http://dx.doi.org/10.1088/1751-8113/42/9/095302},
   DOI={10.1088/1751-8113/42/9/095302},
   number={9},
   journal={Journal of Physics A: Mathematical and Theoretical},
   publisher={IOP Publishing},
   author={Bombin, H and Martin-Delgado, M A},
   year={2009},
   month=Feb, pages={095302} }

@misc{cowtan2026nearoptimal,
      title={Near-optimal high-rate surgery from linear PCPPs}, 
      author={Alexander Cowtan and Benjamin Ide},
      year={2026},
      eprint={2609.26973},
      archivePrefix={arXiv},
      primaryClass={quant-ph},
      url={https://arxiv.org/abs/2609.26973}, 
}
\end{document}